\documentclass[11pt]{article}
\usepackage[margin=1in]{geometry}
\usepackage{hyperref}
\usepackage{amsmath}
\usepackage{amssymb}
\usepackage{amsfonts}
\usepackage{amsthm}
\usepackage{thm-restate}
\usepackage{subcaption}
\usepackage{xcolor}
\usepackage{graphicx}
\usepackage{cleveref}
\usepackage{bbm}

\newtheorem{definition}{Definition}[section]
\newtheorem{theorem}{Theorem}[section]
\newtheorem{corollary}{Corollary}[section]
\newtheorem{lemma}{Lemma}[section]

\Crefname{lemma}{Lemma}{Lemmas}

\title{Parameterized Algorithms for String Matching to DAGs: Funnels and Beyond\thanks{This work was partially funded by the Academy of Finland
(grants No. 352821, 328877). I am very grateful to Alexandru~I.~Tomescu for initial discussions on funnel algorithms, to Veli Mäkinen for discussions on applying KMP in DAGs, to Massimo Equi and Nicola Rizzo for the useful discussions, and to the anonymous reviewers for their useful comments.}}

\author{Manuel C\'aceres\thanks{Department of Computer Science, University of Helsinki, Finland, \texttt{manuel.caceresreyes@helsinki.fi}.}}

\date{}

\newcommand{\paths}{\mathcal{P}}
\newcommand{\source}{\mathcal{S}}
\newcommand{\T}{\mathcal{T}}
\newcommand{\ST}{\mathcal{ST}}

\begin{document}

\maketitle

\begin{abstract}
The problem of String Matching to Labeled Graphs (SMLG) asks to find all the paths in a labeled graph $G = (V, E)$ whose spellings match that of an input string $S \in \Sigma^m$. SMLG can be solved in quadratic $O(m|E|)$ time~[Amir et~al., JALG], which was proven to be optimal by a recent lower bound conditioned on SETH~[Equi et~al., ICALP 2019]. The lower bound states that no strongly subquadratic time algorithm exists, even if restricted to directed acyclic graphs (DAGs).

In this work we present the first parameterized algorithms for SMLG in DAGs. Our parameters capture the topological structure of $G$. All our results are derived from a generalization of the Knuth-Morris-Pratt algorithm~[Park and Kim, CPM 1995] optimized to work in time proportional to the number of prefix-incomparable matches.

To obtain the parameterization in the topological structure of $G$, we first study a special class of DAGs called funnels~[Millani et~al., JCO] and generalize them to $k$-funnels and the class $\ST_k$. We present several novel characterizations and algorithmic contributions on both funnels and their generalizations.
\end{abstract}

\section{Introduction}

Given a labeled graph $G = (V, E)$ (vertices labeled with characters) and a string $S$ of length $m$ over an alphabet $\Sigma$ of size $\sigma$, the problem of \emph{String Matching to Labeled Graph (SMLG)} asks to find all paths of $G$ spelling $S$ in their characters; such paths are known as \emph{occurrences} or \emph{matches} of $S$ in $G$. This problem is a generalization of the classical \emph{string matching (SM)} to a text $T$ of length $n$, which can be encoded as an SMLG instance with a path labeled with $T$. Labeled graphs are present in many areas such as information retrieval~\cite{conklin1987hypertext,nielsen1990hypertext,baeza1999modern}, graph databases~\cite{angles2008survey,angles2017foundations,perez2009semantics,barcelo2013querying} and bioinformatics~\cite{computational2018computational}, SMLG being a primitive operation to locate information on them.

It is a textbook result~\cite{aho1974design,cormen2022introduction,sedgewick2011algorithms} that the classical SM can be solved in linear $O(n+m)$ time. For example, the well-known \emph{Knuth-Morris-Pratt} algorithm (KMP)~\cite{knuth1977fast} preprocess $S$ and then scans $T$ while maintaining the longest matching prefix of $S$. However, for SMLG a recent result~\cite{backurs2016regular,equi2019complexity} shows that there is no strongly subquadratic $O(m^{1-\epsilon}|E|), O(m|E|^{1-\epsilon})$ time algorithm unless the \emph{strong exponential time hypothesis (SETH)} fails, and the most efficient current solutions~\cite{amir2000pattern,navarro2000improved,rautiainen2017aligning,jain2020complexity} match this bound, thus being optimal in this sense. Moreover, these algorithms solve the approximate version of SMLG (errors in $S$ only) showing that both problems are equally hard under SETH, which is not the case for SM~\cite{DBLP:journals/siamcomp/BackursI18}.\\

\par{\textbf{The history of (exact) SMLG.}} SMLG can be traced back to the publications of Manber and Wu~\cite{manber1992approximate} and Dubiner et~al.~\cite{dubiner1994faster} where the problem is defined for the first time, and solved in linear time on directed trees by using an extension of KMP. Later Akutsu~\cite{akutsu1993linear} used a sampling on $V$ and a suffix tree of $S$ to solve the problem on (undirected) trees in linear time and Park and Kim~\cite{park1995string} obtained a $O(N + m|E|)$\footnote{$N$ is the total length of the labels in $G$ in a more general version of the problem where vertices are labeled with strings.} time algorithm for directed acyclic graphs (DAGs) by extending KMP on a topological ordering of $G$ (we call this the \emph{DAG algorithm}). Finally, Amir et al.~\cite{amir2000pattern} showed an algorithm with the same running time for general graphs with a simple and elegant idea that was later used to solve the approximate version~\cite{rautiainen2017aligning,jain2020complexity}, and that has been recently generalized as the labeled product~\cite{rizzo2021labeled}. The stricter lower bound of Equi et~al.~\cite{equi2019complexity} shows that the problem remains quadratic (under SETH) even if the problem is restricted to deterministic DAGs with vertices of two kinds: indegree at most $1$ and outdegree $2$, and indegree $2$ and outdegree at most $1$~\cite[Theorem 1]{equi2019complexity}, or if restricted to undirected graphs with degree at most $2$~\cite[Theorem 2]{equi2019complexity}. Furthermore, they show how to solve the remaining cases (in/out-trees whose roots can be connected by a cycle) in linear time by an extension of KMP. Later they showed~\cite{equi2021graphs} that the quadratic lower bound holds even when allowing polynomial indexing time.\\

\par{\textbf{An (important) special case.}} Gagie et~al.~\cite{gagie2017wheeler} introduced \emph{wheeler graphs} as a generalization of prefix sortable techniques~\cite{ferragina2000opportunistic,grossi2000compressed,ferragina2005structuring,bowe2012succinct,siren2014indexing} applied to labeled graphs. On wheeler graphs, SMLG can be solved in time $O(m\log{|E|})$~\cite{gagie2017wheeler} after indexing, however, it was shown that the languages recognized by wheeler graphs (intuitively the set of strings they encode), is very restrictive~\cite{alanko2020regular,alanko2021wheeler}. Latter, Cotumaccio and Prezza~\cite{cotumaccio2021indexing} generalized wheeler graphs to $p$\emph{-sortable graphs}, capturing every labeled graph by using the parameter $p$: the minimum width of a colex relation over the vertices of the graph. On $p$-sortable graphs, SMLG can be solved in time $O(mp^2\log{(p\sigma}))$ after indexing, however, the problems of deciding if a labeled graph is wheeler or $p$-sortable are NP-hard~\cite{DBLP:conf/esa/GibneyT19}. In a recent work, Cotumaccio~\cite{cotumaccio2022graphs} defined \emph{$q$-sortable graphs} as a relaxation of $p$-sortable ($q < p$), which can be indexed in $O(|E|^2+|V|^{5/2})$ time but still solve SMLG in time $O(mq^2\log{(q\sigma}))$.
 
\subsection{Our results}
We present parameterized algorithms for SMLG in DAGs. Our parameters capture the topological structure of $G$. These results are related to the line of research ``FPT inside P''~\cite{giannopoulou2017polynomial,caceres2021safety,fomin2018fully,koana2021data,abboud2016approximation,caceres2021a,caceres2022sparsifying,caceres2022minimum,makinen2019sparse,ma2022graphchainer} of finding parameterizations for polynomially-solvable problems.

All our results are derived from a new version of the DAG algorithm~\cite{park1995string}, which we present in \Cref{sec:dag-incomparable}. Our algorithm is optimized to only carry \emph{prefix-incomparable} matches (\Cref{def:prefix-incomparable}) and process them in time proportional to their size (\Cref{lemma:parameterized-vertex} further optimized in \Cref{cor:sorting-parameterized-vertex,cor:linear-parameterized-vertex}). Prefix-incomparable sets suffice to capture all prefix matches of $S$ ending in a vertex $v$ (\Cref{def:bv-piv}). By noting that the size of prefix-incomparable sets is upper-bounded by the structure of $S$ (\Cref{lemma:bounded-size-prefix-incomparable} in \Cref{sec:the-pattern}), we obtain a parameterized algorithm (\Cref{thm:dag-algorithm-pattern} in \Cref{sec:the-pattern}) that beats the DAG algorithm on periodic strings.

To obtain the parameterization on the topological structure of the graph we first study and generalize a special class of DAGs called \emph{funnels} in \Cref{sec:funnels-and-beyond}. \\

\par{\textbf{Funnels.}} Funnels are DAGs whose source-to-sink paths contain a \emph{private} edge that is not used by any other source-to-sink-path. Although more complex that in/out-forests, its simplicity has allowed to efficiently solve problems that remain hard even when the input is restricted to DAGs, including: DAG partitioning~\cite{millani2020efficient}, $k$-linkage~\cite{millani2020efficient}, minimum flow decomposition~\cite{khan2022safety,khan2022improving,khan2022optimizing}, a variation of network inhibition~\cite{lehmann2017thecomp} and SMLG (this work). Millani et~al.~\cite{millani2020efficient} showed that funnels can also be characterized by a partition into an in-forest plus an out-forest (the \emph{vertex partition} characterization), or by the absence of certain forbidden paths (the \emph{forbidden path} characterization), and propose how to find a minimal forbidden path in quadratic $O(|V|(|V|+|E|))$ time and a recognition algorithm running in $O(|V|+|E|)$ time. They used the latter to develop branching algorithms for the NP-hard problems of vertex and edge deletion distance to a funnel, obtaining a fix-parameter quadratic solution. Analogous to the minimum feedback set problem~\cite{karp1972reducibility}, the vertex (edge) deletion distance to a funnel problem asks to find the minimum number of vertices (edges) that need to be removed from a graph so that the resulting graph is a funnel.\\

We propose three (new) linear time recognition algorithms of funnels (\Cref{sec:three-recognition}), each based on a different characterization, improving the running time of the branching algorithm to parameterized linear time (see \Cref{sec:linear-distance}). We generalize funnels to $k$-funnels by allowing private edges to be shared by at most $k$ source-to-sink paths (\Cref{def:kprivate,def:kfunnel}). We show how to recognize them in linear time (\Cref{thm:k-funnel-linear-recognition}) and find the minimum $k$ for which a DAG is a $k$-funnel (\Cref{cor:exponential-min-k,lemma:linear-min-k}). We then further generalize $k$-funnels to the class of DAGs $\ST_k$ (\Cref{def:stk,lemma:subset-stk}), which (unlike $k$-funnels for $k>1$, see \Cref{fig:strict-containment}) can be characterized (and efficiently recognized, see \Cref{thm:stk-partitioning}) by a partition into a graph of the class $\source_k$ (generalization of out-forest, see \Cref{def:sk-tk}) and a graph of the class $\T_k$ (generalization of in-forest, see \Cref{def:sk-tk}).

We obtain our parameterized results in \Cref{sec:the-dag} by noting that, analogous to the fact that in KMP we only need the longest prefix match, in the DAG algorithm we can bound the size of the prefix-incomparable sets by the number of paths from a source or the number of paths to a sink, $\mu_s(v)$ and $\mu_t(v)$, respectively (\Cref{lemma:prefix-incomparable-topo-bound}).

\begin{restatable}{theorem}{paramAlgTopoOne}\label{thm:param-alg-top-1}
 Let $G = (V, E)$ be a DAG, $\Sigma$ a finite ($\sigma = |\Sigma|$) alphabet, $\ell: V \rightarrow \Sigma$ a labeling function and $S \in \Sigma^m$ a string. We can decide whether $S$ has a match in $G,\ell$ in time $O((|V|+|E|)k + \sigma m)$, where $k = \min(\max_{v\in V} \mu_s(v), \max_{v\in V} \mu_t(v))$.
 \end{restatable}

In particular, this implies linear time algorithms for out-forests and in-forests, and for every DAG in $\source_k$ or $\T_k$ for constant $k$. Finally, we solve the problem on DAGs in $\ST_k$ (thus also in $k$-funnels), by using the vertex partition characterization of $\ST_k$ (\Cref{thm:stk-partitioning}), solving the matches in each part separately with \Cref{thm:param-alg-top-1}, and resolve the matches crossing from one part to the other with a precomputed data structure (\Cref{lemma:prefix-suffix}).

\begin{restatable}{theorem}{paramAlgTopoTwo}\label{thm:param-alg-top-2}
 Let $G = (V, E)$ be a DAG, $\ell: V \rightarrow \Sigma$ a labeling function and $S \in \Sigma^m$ a string. We can decide whether $S$ has a match in $G,\ell$ in time $O((|V|+|E|)k^2 + m^2)$, where $k = \max_{v\in V}(\min(\mu_s(v), \mu_t(v)))$.
 \end{restatable}

 \section{Preliminaries}
 We work with a (directed) graph $G = (V, E)$, a function $\ell: V \rightarrow \Sigma$ labeling the vertices of $G$ with characters from a finite alphabet $\Sigma$ of size $\sigma$, and a sequence $S[1..m] \in \Sigma^m$. \\
 
  \par{\textbf{Graphs.}} A graph $S = (V_S, E_S)$ is said to be a \emph{subgraph} of $G$ if $V_S \subseteq V$ and $E_S \subseteq E$. If $V' \subseteq V$, then $G[V']$ is the subgraph \emph{induced by} $V'$, defined as $G[V'] = (V', \{(u, v) \in E ~:~ u,v \in V'\})$. We denote $G^r = (V, E^r)$ to be the \emph{reverse} of $G$ ($E^r = \{(v, u) \mid (u, v) \in E\}$). For a vertex $v\in V$ we denote by $N^{-}_{v}$ ($N^{+}_{v}$) the set of \emph{in(out)-neighbors} of $v$, and by $d^{-}_{v} = |N^{-}_{v}|$ ($d^{+}_{v} = |N^{+}_{v}|$) its \emph{in(out)degree}. A \emph{source (sink)} is a vertex with zero in(out)degree. \emph{Edge contraction} of $(u, v) \in E$ is the graph operation that removes $(u, v)$ and merges $u$ and $v$. A \emph{path} $P$ is a sequence $v_{1},\ldots,v_{|P|}$ of different vertices of $V$ such that $(v_{i}, v_{i+1}) \in E$ for every $i \in \{1,\ldots, |P|-1\}$. We say that $P$ is \emph{proper} if $|P| \ge 2$, a \emph{cycle} if $(v_{|P|}, v_{1}) \in E$, and \emph{source-to-sink} if $v_1$ is  source and $v_{|P|}$ is a sink. We say that $u\in V$ ($e \in E$) \emph{reaches} $v\in V$ ($f \in E$) if there is a path from $u$ (the head of $e$) to $v$ (the tail of $f$). If $G$ does not have cycles it is called \emph{directed acyclic graph} (DAG). A \emph{topological ordering} of a DAG is a total order $v_{1}, \ldots, v_{|V|}$ of $V$ such that for every $(v_{i}, v_{j}) \in E$, $i < j$. It is known~\cite{kahn1962topological,tarjan1976edge} how to compute a topological ordering in $O(|V|+|E|)$ time, and we assume one ($v_{1}, \ldots , v_{|V|}$) is already computed if $G$ is a DAG\footnote{Our algorithms run in $\Omega(|V|+|E|)$ time.}. An \emph{out(in)-forest} is a DAG such that every vertex has in(out)degree at most one, if it has a unique source (sink) it is called an \emph{out(in)-tree}. The \emph{label} of a path $P = v_{1},\ldots,v_{|P|}$ is the sequence of the labels of its vertices, $\ell(P) = \ell(v_{1})\ldots\ell(v_{|P|})$. \\
 
 \par{\textbf{Strings.}} We say that $S$ has a \emph{match} in $G,\ell$ if there is a path whose label is equal to $S$, every such path is an \emph{occurrence} of $S$ in $G,\ell$. We denote $S[i..j]$ (also $S[i]$ if $i=j$, and the empty string if $j<i$) to be the segment of $S$ between position $i$ and $j$ (both inclusive), we say that it is \emph{proper} if $i > 1$ or $j<m$, a \emph{prefix} if $i = 1$ and a 
 \emph{suffix} if $j = m$. We denote $S^r$ to be the reverse of $S$ ($S^r[i] = S[m-i+1]$ for $i \in \{1,\ldots, m\}$). A segment of $S$ is called a \emph{border} if it is a proper prefix and a proper suffix at the same time. The \emph{failure function} of $S$, $f_{S}: \{1,\ldots, m\} \rightarrow \{0,\ldots, m\}$ (just $f$ if $S$ is clear from the context), is such that $f_{S}(i)$ is the length of the longest border of $S[1..i]$. We also use $f_{S}$ to denote the in-tree $(\{0,\ldots, m\}, \{(i, f_{S}(i)) \mid i \in \{1,\ldots, m\}\})$, also known as the \emph{failure tree}~\cite{10.5555/314464.314675} of $S$. By definition, the lengths of all borders of $S[1..i]$ in decreasing order are  $f_{S}(i), f^2_{S}(i), \ldots, 0$. The matching automaton of $S$, $A_S: \{0, \ldots, m\} \times \Sigma \rightarrow \{0, \ldots, m\}$, is such that $A_S(i, a)$ is the length of the longest border of $S[1..i]\cdot a$.
 It is known how to compute $f_{S}$ in time $O(m)$~\cite{knuth1977fast} and $A_{S}$ in time $O(\sigma m)$~\cite{aho1974design,cormen2022introduction,sedgewick2011algorithms}, and we assume they are already computed\footnote{Our algorithms run in $\Omega(\sigma m)$ time.}.
 
 \section{The DAG algorithm on prefix-incomparable matches}\label{sec:dag-incomparable}
 
 A key idea in our linear time parameterized algorithm is that of prefix-incomparable sets of the string $S$. We will show that one prefix-incomparable set per vertex suffices to capture all the matching information. See \Cref{fig:prefix-incomparable} for an example of these concepts.

 \begin{definition}[Prefix-incomparable]\label{def:prefix-incomparable}
 Let $S \in \Sigma^m$ be a string. We say that $i < j \in \{0, \ldots, m\}$ are prefix-incomparable (for $S$) if $S[1..i]$ is not a border of $S[1..j]$. We say that $B \subseteq \{0,\ldots, m\}$ is prefix-incomparable (for $S$) if for every $i < j\in B$, $i$ and $j$ are prefix-incomparable (for $S$).
 \end{definition}
 
  \begin{figure}
      \centering
      \includegraphics{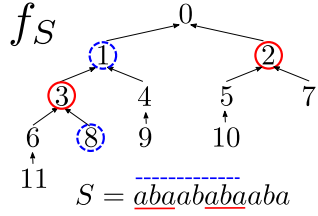}
      \caption{A string $S = abaababaaba$ and its failure tree $f_S$, with $w = |\{i \in \{0,\ldots, 11\} \mid \not \exists j, f_S(j) = i\}| = |\{7,8,9,10,11\}| = 5$. On the string it is shown in segmented (blue) and solid (red) lines that prefix $S[1..3] = aba$ is a border of prefix $S[1..8] = abaababa$, which can also be seen in the tree since $3$ is an ancestor (parent in this case) of $8$. In the tree two sets are shown $B_1 = \{1, 8\}$ in segmented (blue) circles, which is prefix-comparable, and $B_2 = \{2, 3\}$ in solid (red) circles, which is prefix-incomparable. If $B_1\cup B_2 \cup \{0\}$ is $B_v$ for some $v \in V$, then $PI_v = \{2,8\}$.}
      \label{fig:prefix-incomparable}
  \end{figure}

 In our algorithm we will compute for each vertex $v$ a prefix-incomparable set representing all the prefixes of $S$ that match with a path ending in $v$. More precisely, if $B_v$ is the set of all the prefixes of $S$ that match with a path ending in $v$, then the algorithm will compute $PI_v \subseteq B_v$ such that $PI_v$ is prefix-incomparable and for every $i \in B_v$ there is a $j \in PI_v$ such that $i$ is ancestor of $j$. Note that such a set always exists and it is unique, it corresponds to the leaves of $f_S[B_v]$. To obtain a linear time parameterized algorithm we show how to compute $PI_v$ from the sets $PI_u$, $u \in N^{-}_{v}$, in time parameterized by the size of these sets.
 
 \begin{definition}[$B_v, PI_v$]\label{def:bv-piv}
 Let $G = (V, E)$ be a DAG, $\ell: V \rightarrow \Sigma$ a labeling function, and $S \in \Sigma^m$ a string. For every $v \in V$ we define the sets:
 \begin{itemize}
     \item $B_v = \{i \in \{0,\ldots,m\} \mid \exists P \text{ path of } G\text{ ending in $v$ and } \ell(P) = S[1..i]\}$\footnote{Here we consider that the empty path always exists and its label is the empty string, thus $0\in B_v$.}
     \item $PI_v \subseteq B_v$ as the unique prefix-incomparable set such that for every $i \in B_v$ there is a $j \in PI_v$ such that $i = j$ or $S[1..i]$ is a border of $S[1..j]$
 \end{itemize}
 \end{definition}
 
 \begin{lemma}\label{lemma:parameterized-vertex}
 Let $G = (V, E)$ be a DAG, $v \in V$, $S \in \Sigma^m$ a string, $f_S$ its failure tree and $A_S$ its matching automaton. We can compute $PI_v$ from $PI_u$ for every $u \in N^{-}_{v}$ in time $O\left(w^2\cdot d^{-}_{v}\right)$ or in time $O\left(\left(k_v := \sum_{u \in N^{-}_{v}} |PI_u|\right)^2\right)$, after $O(m)$ preprocessing time.
 \end{lemma}
 \begin{proof}
  We precompute constant time lowest common ancestor ($LCA$) queries~\cite{aho1976finding} of $f_S$ in $O(m)$ time~\cite{gabow1985linear,schieber1988finding,berkman1993recursive,bender2000lca,bender2005lowest,alstrup2004nearest,fischer2006theoretical}. Note that with this structure we can check whether $i < j$ are prefix-incomparable in constant time ($LCA(i, j) < i$). 
  
  If $v$ is a source we have that either $B_v = PI_v = \{0\}$ if $\ell(v) \neq S[1]$ or $PI_v = \{1\}$ if $\ell(v) = S[1]$, otherwise we proceed as follows. To obtain the $O\left(k_v^2\right)$ time, we first append all the elements of every $PI_u$ for $u\in N^{-}_{v}$ into a list $\mathcal{L}$ (of size $k_v$), then we replace every $i \in \mathcal{L}$ by $A_S(i, \ell(v))$, and finally we check (at most) every pair $i < j$ of elements of $\mathcal{L}$ and test (in constant time) if they are prefix-incomparable, if they are not we remove $i$ from the list. After these $O(|\mathcal{L}|^2) = O(k_v^2)$ tests $\mathcal{L} = PI_v$. 
  
  To obtain the $O\left(w^2\cdot d^{-}_{v}\right)$ time, we process the in-neighbors of $v$ one by one and maintain a prefix-incomparable set representing the prefix matches incoming from the already processed in-neighbors. That is, we maintain a prefix-incomparable set $PI'$, and when we process the next in-neighbor $u \in N^{-}_{v}$ we append all elements of $PI'$ and $\{A_S(i,\ell(v)) \mid i \in PI_u\}$ into a list $\mathcal{L}'$ of size $O(w)$ (by \Cref{lemma:bounded-size-prefix-incomparable}), then we use the same quadratic procedure applied on $\mathcal{L}$ in time $O(w^2)$ to obtain the new $PI'$. After processing all in-neighbors in time $O(d^{-}_{v}\cdot w^2)$, we have $PI' = PI_v$. Next, we show the correctness of both procedures.
 
 Let $PI'_v$ be the result after applying some of the procedures explained before (that is the final state of $PI'$ or $\mathcal{L}$), by construction $PI'_v$ is prefix-incomparable. Now, consider $i \in B_v$ and a path $P$ ending in $v$ with $\ell(P) = S[1..i]$. If $P$ is of length zero (only one vertex), then $i = 1$ and $\ell(P) = \ell(v) = S[1]$. Consider any $u\in N^{-}_{v}$, and any $j \in PI_u$ (there is at least one since $0 \in B_u)$, this value is mapped to $A_S(j, \ell(v)) = A_S(j, S[1]) = j'$. By definition of the matching automaton, $j'$ is the longest border of $S[1..j]\cdot S[1]$, thus $S[1]$ is a border of $S[1..j']$ or $j' = 1$, in both procedures $j'$ can only be removed from $PI'_v$ if a longer prefix contains $S[1..j']$ as a border and also $S[1]$. If $P$ is a proper path and $i > 1$, consider the second to last vertex $u$ of $P$. Note that $i-1 \in B_u$, and thus there is $j\in PI_u$, such that $S[1..i-1]$ is a border of $S[1..j]$, this value is mapped to $j' = A_S(j, \ell(v)) = A_S(j, S[i])$, which is the length of the longest border of $S[1..j]\cdot S[i]$, but since $S[1..i]$ is also a border of $S[1..j]\cdot S[i]$, then $S[1..i]$ is also a border of $S[1..j']$ or $j' = i$. Again, in both procedures $j'$ can only be removed from $PI'_v$ if a longer prefix contains $S[1..j']$ as a border and also $S[1..i]$. Finally, we note that $PI'_v \subseteq B_v$ since every $i \in PI'_v$ corresponds to a match of $S[1..i]$ by construction.
 \end{proof}
 
 We improve the dependency on $w$ and $k_v$ by replacing the quadratic comparison by sorting plus a linear time algorithm on the balanced parenthesis representation~\cite{jacobson1989space,munro2001succinct} of $f_S$.
 
 \begin{restatable}{lemma}{sortingParameterizedVertex}\label{cor:sorting-parameterized-vertex}
 We can obtain \Cref{lemma:parameterized-vertex} in time $O(sort(w, m)\cdot d^{-}_{v})$ or in time $O(sort(k_v, m))$, where $sort(n, p)$ is the time spent by an algorithm sorting $n$ integers in the range $\{0, \ldots, p\}$.
 \end{restatable}
 \begin{proof}
  We compute the balanced parenthesis (BP)~\cite{jacobson1989space,munro2001succinct} representation of the topology of $f_S$, that is, we traverse $f_S$ from the root in preorder, appending an open parenthesis when we first arrive at a vertex, and a closing one when we leave its subtree. As a result we obtain a balanced parenthesis sequence of length $2(m+1)$, where every vertex $i\in f_S$ is mapped to its open parenthesis position $open[i]$ and to its close parenthesis position $close[i]$, which can be computed and stored at preprocessing time. Note that in this representation, $i$ is ancestor of $j$ (and thus prefix-comparable) if and only if $open[i] \le open[j] \le close[i]$. As such, if we have a list of $O(k_v)$ (or $O(w)$) ($\mathcal{L}$ and $\mathcal{L}'$ from \Cref{lemma:parameterized-vertex}) values, we can compute the corresponding prefix-incomparable sets as follows.
  
  First, we sort the list by increasing $open$ value, this can be done in $O(sort(k_v, m))$ (or $O(sort(w, m)$), since this sorting is equivalent to sort by increasing $open/2 \in \{0, \ldots, m\}$ value. Then, we process the list in the sorted order, if two consecutive values $i$ and $j$ in the order are prefix-comparable (that is, if $open[j] \le close[i]$) then we remove $i$ and continue to the next value $j$. At the end of this $O(k_v)$ (or $O(w)$) time processing we obtain the desired prefix-incomparable set.
 \end{proof}

 If we use techniques for integer sorting~\cite{van1975preserving,willard1983log,kirkpatrick1983upper} we can get $O(k_v\log\log{m})$ (or $O(w\log\log{m})$) time for sorting, however introducing $m$ into the running time. We can solve this issue by using more advanced techniques~\cite{han1995conservative,andersson1998sorting,han2002deterministic} obtaining a $O(k_v\log\log{k_v})$ (or $O(w\log\log{w})$) time for sorting. However, we show that by using the suffix-tree~\cite{weiner1973linear,mccreight1976space,ukkonen1995line} of $S^r$ we can obtain a linear dependency on $w$ and $k_v$.
 
 \begin{restatable}{theorem}{linearParameterizedVertex}\label{cor:linear-parameterized-vertex}
 We can obtain the result of \Cref{lemma:parameterized-vertex} in time $O(w\cdot d^{-}_{v})$ or in time $O(k_v)$.
 \end{restatable}
 \begin{proof}
 We reuse the procedure of \Cref{cor:sorting-parameterized-vertex} but this time on the BP representation of the topology of the suffix-tree $T_r$ of $S^r$, which has $O(m)$ vertices and can be built in $O(m)$ time~\cite{weiner1973linear,mccreight1976space,ukkonen1995line}. Note that every suffix represented in $T_r$ corresponds to a prefix of $S$ (spelled in the reverse direction). Moreover, $i \le j$ are prefix-comparable if and only if the vertex representing $i$ in $T_r$ ($T_r[i]$) is a ancestor of $T_r[j]$, the same property as in $f_S$. Furthermore, if $B$ is prefix-incomparable and $A(j, a) = j+1$ for every $j\in B$, then the positions of the vertices in $A(B, a)$ in $T_r$ follow the same order as the ones in $B$, since the suffix-tree is lexicographically sorted.
 
 Now, we show how to obtain the prefix-incomparable set representing $A(PI_u, \ell(v))$ in $|PI_u|$ time assuming that $PI_u$ is sorted by increasing ($open$) position in $T_r$.
 
 We first separate $PI_u$ into the elements $i \in M$  with $S[i+1] = \ell(v)$ and $i \in E$ with $S[i+1] \ne \ell(v)$ (in the same relative order as in $PI_u$, which is supposed to be in increasing order). Since $M$ is prefix-incomparable the positions of the vertices in $A(M, \ell(v))$ in $T_r$ follow the same order as the ones in $M$. We then obtain the list $E_u$ by applying $T_r[A(i, \ell(v)) - 1]$ for every $i \in E$ (if $A(i, \ell(v)) = 0$ we do not process $i$), and then for any pair of consecutive elements $x$ before $y$ in $E_u$ such that $y\le x$ we remove $y$ from $E_u$, and repeat this until no further such inversion remains, thus obtaining an increasing list in $E_u$ representing vertices in $T_r$. Next, since $E_u$ is sorted we can obtain the list $PI_E$ of prefix-incomparable elements representing $E_u$, and finally apply $A(PI_E, \ell(v))$ (which also follows an increasing order in $T_r$), merge it with $A(M, \ell(v))$, and compute the prefix-comparable elements of this merge.
 
 The correctness of the previous procedure follows by the fact that if there is an inversion $y < x$ in $E_u$, then the prefix $A(j,\ell(v))-1$ represented by $y$ in $T_r$ is a border of the prefix $A(i,\ell(v))-1$ represented by $x$ (and thus is safe to remove $y$). For this, first note that $i$ appears before $j$ in $E$, then $S[i..1] <_{lex} S[j..1]$, and since $i$ is prefix-incomparable with $j$ there is a $k\ge 1$ such that $S[i..i+k-1] = S[j..j+k-1]$ and $S[i+k] <_{lex} S[j+k]$. Then, since $y$  appears before $x$ in $E_u$, then $S[A(j,\ell(v))-1..1] <_{lex} S[A(i,\ell(v))-1..1]$, but since $A(i,\ell(v))-1$ is a border of $i$ and $A(j,\ell(v))-1$ is a border of $j$, $S[A(j,\ell(v))-1..1]$ must be a prefix of $S[A(i,\ell(v))-1..1]$, and thus $S[1..A(j,\ell(v))-1]$ is a border of $S[1..A(j,\ell(v))-1]$.
 
 The corollary is obtained by maintaining the $PI_v$ sets sorted by position in $T_r$, and noting that the previous procedure runs in linear $O(|P_u|)$ time.
 \end{proof}

 In \Cref{sec:the-pattern} we show how to use \Cref{cor:linear-parameterized-vertex} to derive a parameterized algorithm using parameter $w = |\{i \in \{0,\ldots, m\} \mid \not \exists j, f_S(j) = i\}|$, improving on the classical DAG algorithm when $S$ is a periodic string. Next, we will present our results on recognizing funnels and their generalization (\Cref{sec:funnels-and-beyond}), and how to use these classes of graphs and \Cref{cor:linear-parameterized-vertex} to obtain parameterized algorithms using parameters related to the topology of the DAG (\Cref{sec:the-dag}).
 
 \section{Funnels and beyond}\label{sec:funnels-and-beyond}
 
 Recall that funnels are DAGs whose source-to-sink paths have at least one \emph{private} edge\footnote{For the sake of simplicity, we assume that there are no isolated vertices, thus any source-to-sink path has at least one edge.}, that is, an edge used by only one source-to-sink path. More formally, 
 
 \begin{definition}[Private edge]
 Let $G = (V, E)$ be a DAG and $\paths$ the set of source-to-sink paths of $G$. We say that $e\in E$ is \emph{private} if $\mu(e) := |\{P \in \paths \mid e \in P\}| = 1$. If $\mu(e) > 1$, we say that $e$ is \emph{shared}.
 \end{definition}

 \begin{definition}[Funnel]\label{def:funnel}
 Let $G = (V, E)$ be a DAG and $\paths$ the set of source-to-sink paths of $G$. We say that $G$ is a funnel if for every $P \in \paths$ there exists $e \in P$ such that $e$ is private.
 \end{definition}
 
 Millani et al.~\cite{millani2020efficient} showed two other characterizations of funnels.
 \begin{theorem}[\cite{millani2020efficient}]\label{thm:funnels-characterizations}
 Let $G = (V, E)$ be a DAG. The following are equivalent:
 \begin{enumerate}
     \item $G$ is a funnel
     \item There exists a partition $V = V_1 \dot\cup V_2$ such that $G[V_1]$ is an out-forest, $G[V_2]$ is an in-forest and there are no edges from $V_2$ to $V_1$
     \item There is no path $P$ such that its first vertex has more than one in-neighbor (a merging vertex) and its last vertex more than one out-neighbor (a forking vertex). Such a path is called forbidden
 \end{enumerate}
 \end{theorem}
 
They also gave a $O(|V|+|E|)$ time algorithm to recognize whether a DAG $G$ is a funnel, and a $O(|V|(|V|+|E|))$ time algorithm to find a minimal forbidden path in a general graph, that is, a forbidden path that is not contained in another forbidden path.
 
 \subsection{Three (new) linear time recognition algorithms}\label{sec:three-recognition}

 We first show how to find a minimal forbidden path in time $O(|V|+|E|)$ in general graphs, improving on the quadratic algorithm of Millani et al.~\cite{millani2020efficient}.
 
 \begin{restatable}{lemma}{unitig}\label{lemma:unitig}
 Let $G = (V, E)$ be a graph. In $O(|V|+|E|)$ time, we can decide if $G$ contains a forbidden path, and if one exists we report a minimal forbidden path.
 \end{restatable}
 \begin{proof}
 In the bioinformatics community minimal forbidden paths are a subset of \emph{unitigs} and it is well known how to compute them in $O(|V|+|E|)$ time (see e.g.~\cite{kececioglu1995combinatorial,jackson2009parallel,kingsford2010assembly,medvedev2007computability}), here we include a simple algorithm for completeness. We first compute the indegree and outdegree of each vertex and check whether exists a forbidden path of length zero or one, all in $O(|V|+|E|)$ time, in the process we also mark all vertices except the ones with unit indegree and outdegree. If no path is found we iterate over the vertices one last time. If the current vertex is not marked we extend it back and forth as long as there is a unique extension and mark the vertices in this extension, finally we check whether the first vertex is merging and the last forking. This last iteration takes $O(|V|)$ time.
 \end{proof}

 \Cref{lemma:unitig} provides our first linear time recognition algorithm and, as opposed to the algorithm of Millani et al.~\cite{millani2020efficient}, it also reports a minimal forbidden path given a general graph. Moreover, in \Cref{sec:linear-distance}, we show that \Cref{lemma:unitig} provides a linear time parameterized algorithm for the NP-hard (and inapproximable) problem of deletion distance of a general graph to a funnel~\cite{lund1993approximation,millani2020efficient}. Millani et al.~\cite{millani2020efficient} solved this problem in (parameterized) quadratic time and in (parameterized) linear time only if the input graph is a DAG.
 
 Next, we show another linear time recognition algorithm, which additionally finds the partition $V = V_{1} \dot\cup V_{2}$ from \Cref{thm:funnels-characterizations}. Finding such a partition will be essential for our solution to SMLG. From now we will assume that the input graph is a DAG since this condition can be checked in linear time~\cite{kahn1962topological,tarjan1976edge}.
 
 \begin{lemma}
 Let $G = (V, E)$ be a DAG. We can decide in $O(|V|+|E|)$ time whether $G$ is a funnel. Additionally, if $G$ is a funnel, the algorithm reports a partition $V = V_{1} \dot\cup V_{2}$ such that $G[V_1]$ is an out-forest, $G[V_2]$ is an in-forest and there are no edges from $V_2$ to $V_1$.
 \end{lemma}
 \begin{proof}
 We start a special BFS traversal from all the source vertices of $G$. The traversal only adds vertices to the BFS queue if they have not been previously visited (as a typical BFS traversal) and if its indegree is at most one. After the search we define the partition $V_1$ as the set of vertices visited during the traversal and $V_2 = V\setminus V_1$. Finally, we report the previous partition if there are no edges from $V_2$ to $V_1$, and if every vertex of $V_2$ has outdegree at most one. All these steps run in time $O(|V|+|E|)$. 
 
 Note that if the algorithm reports a partition, then this satisfies the required conditions to be a funnel ($G[V_{1}]$ is an out-forest since every vertex visited in the traversal has indegree at most one). Moreover, if $G$ is a funnel, we prove that $V_{2}$ is an in-forest and that there are no edges from $V_{2}$ to $V_{1}$. For the first, suppose by contradiction that there is a vertex $v \in V_{2}$ with $d^{+}_{v} > 1$, since every vertex is reached by some source in a DAG then there is a $u \in V_{2}$ with $d^{-}_{u} > 1$ (a vertex that was not added to the BFS queue) that reaches $v$, implying the existence of a forbidden path in $G$, a contradiction. Finally, there cannot be edges from $V_{2}$ to $V_{1}$ since the indegree (in $G$) of vertices of $V_1$ is at most one and its unique (if any) in-neighbor is also in $V_{1}$ by construction.
 \end{proof}
  
  Next, we present another characterization of funnels based on the structure of private/shared edges of the graph, which can be easily obtained by manipulating the original~\Cref{def:funnel}.
  
  \begin{definition}[Funnel]
  Let $G = (V, E)$ be a DAG. We say that $G$ is a funnel if there is no source-to-sink path using only shared edges.
  \end{definition}
  
  As such, another approach to decide whether a DAG $G$ is a funnel is to compute $\mu(e)$ for every $e\in E$ and then perform a traversal that only uses shared edges. Computing the number of source-to-sink paths containing $e$, that is $\mu(e)$, can be done by multiplying the number of source-to-$e$ paths, $\mu_s(e)$, by the number of $e$-to-sink paths, $\mu_t(e)$, each of which can be computed in $O(|V|+|E|)$ time for all edges. The solution consists of a dynamic programming on a topological order (and reverse topological order) of $G$ with the following recurrences.
  \begin{equation}\label{eq:s2s-counting}
      \begin{aligned}
        \mu_s(e = (u, v)) &= \mu_s(u) = \mathbbm{1}_{d^{-}_{u} = 0} +  \sum_{u' \in N^{-}_{u}} \mu_s((u', u))\\
        \mu_t(e = (u, v)) &= \mu_t(v) = \mathbbm{1}_{d^{+}_{v}=0} +\sum_{v' \in N^{+}_{v}} \mu_t((v, v'))\\\
        \mu(e) &= \mu_s(e)\cdot\mu_t(e)
      \end{aligned}
  \end{equation}
    Where $\mathbbm{1}_A$ is the characteristic function evaluating to $1$ if $A$ is true and to $0$ otherwise. 
    We note that solving the dynamic programs of \Cref{eq:s2s-counting} would take $\Omega(|V||E|)$, since the equations are written in terms of edges. It is simple to see that for every $e = (u, v) \in E, \mu_s(e) = \mu_s(u) \land \mu_t(e) = \mu_t(v)$, thus one can compute the dynamic programs in terms of vertices in $O(|V|+|E|)$ time. By simplicity, we will use this observation implicitly in \Cref{thm:k-funnel-linear-recognition,lemma:linear-min-k}.
    The previous algorithm assumes constant time arithmetic operations on numbers up to $\max_{e\in E} \mu(e)$, which can be $O(2^{|V|})$. To avoid this issue, we note that it is not necessary to compute $\mu(e)$, but only to verify that $\mu(e) > 1$. As such, we can recognize shared edges as soon as we identify that $\mu(e) > 1$, that is whenever $\mu_s(e)$ or $\mu_t(e)$ is greater than one in their respective computation. A formal description of this algorithm can be found in \Cref{thm:k-funnel-linear-recognition}.
 
 \subsection{Generalizations of funnels}
 
 To generalize funnels we will allow source-to-sink paths to use only shared edges, but require to have at least one edge shared by at most $k$ different source-to-sink paths.
 
 \begin{definition}[$k$-private edge]\label{def:kprivate}
 Let $G = (V, E)$ be a DAG. We say that $e \in E$ is $k$-private if $\mu(e) \le k$. If $\mu(e) > k$ we say that $e$ is $k$-shared.
 \end{definition}
 
 \begin{definition}[$k$-funnel]\label{def:kfunnel}
 Let $G = (V, E)$ be a DAG. We say that $G$ is a $k$-funnel if there is no source-to-sink path using only $k$-shared edges.
 \end{definition}
 
 The next algorithm is a generalization of the last algorithm in \Cref{sec:three-recognition} to decide if a DAG is a $k$-funnel. It assumes constant time arithmetic operations on numbers up to $k$.
 
 \begin{lemma}\label{thm:k-funnel-linear-recognition}
 We can decide if a DAG $G = (V, E)$ is a $k$-funnel in $O(|V|+|E|)$ time.
 \end{lemma}
 \begin{proof}
 We process the vertices in a topological ordering and use \Cref{eq:s2s-counting} to compute $\mu_s(e)$ in one pass, $\mu_t(e)$ in another pass and $\mu(e)$ in a final pass. To avoid arithmetic operations with numbers greater than $k$, we mark the edges having $\mu_s$ and $\mu_t$ greater than $k$ as $k$-shared during the computations of $\mu_s, \mu_t$. Note that if $\mu_s(e) > k$ or $\mu_t(e)> k$ then $\mu(e) > k$. As such, before computing $\mu_s(e)$ ($\mu_t(e)$) we check if some of the edges from (to) the in(out)-neighbors is marked as $k$-shared. If that is the case we do not compute $\mu_s(e)$ ($\mu_t(e)$) and instead mark $e$ as $k$-shared, otherwise we compute the respective sum of \Cref{eq:s2s-counting}, and if at some point the cumulative sum exceeds $k$ we stop the computation and mark $e$ as $k$-shared. Finally, we find all $k$-shared edges as the marked plus the unmarked with $\mu(e) = \mu_s(e)\cdot\mu_t(e) > k$, perform a traversal only using $k$-shared edges, and report that $G$ is not a $k$-funnel if there is a source-to-sink path using only $k$-shared edges in time $O(|V|+|E|)$.
 \end{proof}
 
 We can use the previous result and exponential search~\cite{bentley1976almost,baeza2010fast} to find the minimum $k$ such that a DAG is a $k$-funnel. Since the exponential search can overshoot $k$ at most by a factor of $2$, this results assumes constant time arithmetic operations on numbers up to $2k$.
 
 \begin{corollary}\label{cor:exponential-min-k}
 Let $G = (V, E)$ be a DAG. We can find the minimum $k$ such that $G$ is a $k$-funnel in $O((|V|+|E|)\log{k})$ time.
 \end{corollary}
 
 Assuming constant time arithmetic operations on numbers up to $\max_{e\in E} \mu(e)$ the problem is solvable in linear time by noting that the answer is equal to the weight of a widest path.
 
 \begin{restatable}{lemma}{linearMinK}\label{lemma:linear-min-k}
 Let $G = (V, E)$ be a DAG. We can find the minimum $k$ such that $G$ is a $k$-funnel in $O(|V|+|E|)$ time.
 \end{restatable}
  \begin{proof}
 We compute $\mu(e)$ for every $e \in E$ by using the dynamic programming algorithm specified by \Cref{eq:s2s-counting} on a topological ordering of $G$. Since constant time arithmetic operations are assumed for numbers up to $\max_{e\in E} \mu(e)$, the previous computation takes linear time. Then, we compute the weight of a source-to-sink path $P$ maximizing $\min_{e\in P} \mu(e)$, and report this value. This problem is known as the widest path problem~\cite{pollack1960maximum,shacham1992multicast,magnanti1993network,ullah2009algorithm,schulze2011new} and it can be solved in linear time in DAGs~\cite{vatinlen2008simple,hartman2012split} by a dynamic program on a topological order of the graph. By completeness, we show a dynamic programming recurrence to compute $W[e]$, the weight of a source-to-$e$ path $P$ maximizing $\min_{e'\in P} \mu(e')$.
 \begin{align*}
     W[e = (u, v)] &= \mu(e)\cdot\mathbbm{1}_{d^{-}_{u} = 0}  + \sum_{u' \in N^{-}_{u}} \min(W[(u', u)], \mu((u', u)))
 \end{align*}
 Finally, note that if we denote $w$ to the weight of a widest path, then there is a source-to-sink path using only $w-1$-shared edges, $G$ is not $w-1$-funnel. Moreover, there cannot be a source-to-sink path using only $w$-shared edges, since such a path would contradict $w$ being the weight of a widest path. As such, $w$ is the minimum $k$ such that $G$ is $k$-funnel.
 \end{proof}
 
 We now define three classes of DAGs closely related to $k$-funnels.
 \begin{definition}\label{def:sk-tk}
 We say that a DAG $G = (V, E)$ belongs to the class $\source_k$ ($\T_k$) if for every $v \in V$, $\mu_s(v)$ ($\mu_t(v)$) $\le k$.
 \end{definition}

 \begin{definition}\label{def:stk}
  We say that a DAG $G = (V, E)$ belongs to the class $\ST_k$ if for every $v \in V$, $\mu_s(v) \le k$ or $\mu_t(v)\le k$.
 \end{definition}
 
 \begin{lemma}\label{lemma:subset-stk}
 $\source_k, \T_k \subseteq k$-funnels $\subseteq \ST_{k}$.
 \end{lemma}
 \begin{proof}
 We first prove that $\source_k, \T_k \subseteq k$-funnels. Consider $G \in \source_k$ ($G \in \T_k$), and take any source-to-sink path $P$ of $G$. Let $(u,v)$ be the last (first) edge of $P$, then by \Cref{eq:s2s-counting} $\mu((u,v)) = \mu_s(u)\cdot\mu_t(v)$, but since $\mu_t(v) = 1$ ($v$ is a sink) and $\mu_s(u) \le k$ ($G \in \source_k$) (analogously, $\mu_s(u) = 1$ and $\mu_t(v) \le k$), then $\mu((u,v)) \le k$, and thus $(u, v)$ is a $k$-private edge.
 To prove that $k$-funnels $\subseteq \ST_{k}$, suppose that $G$ is a $k$-funnel, and by contradiction that there exists $v\in V$ with $\mu_s(v), \mu_t(v) > k$. Consider any source-to-sink path $P$ using $v$. Now, let $(u,w)$ be any edge in $P$ before (after) $v$, then $\mu_t(w) \ge \mu_t(v) > k$ ($\mu_s(u) \ge \mu_s(v) > k$), and thus $\mu((u,w)) = \mu_s(u)\cdot\mu_t(w) > k$. As such, $P$ does not have a $k$-private edge, a contradiction.
 \end{proof}
 
 For $k = 1$, $\source_1$ describes out-forests and $\T_1$ in-forests, thus being more restrictive than funnels. Moreover, we note that the in(out)-star of $k$ vertices, that is $k-1$ vertices pointing to a sink (pointed from a source), $\not\in \source_k$ ($\T_k$), but this graph is a funnel. On the other hand, for the vertex partition characterization of funnels (\Cref{thm:funnels-characterizations}~\cite{millani2020efficient}) we have that $\ST_1 =$ ($1$-)funnels. However, for $k>1$, the containment $k$-funnels $\subseteq \ST_{k}$ is strict (\Cref{fig:strict-containment}).
 
 \begin{figure}
     \centering
     \includegraphics{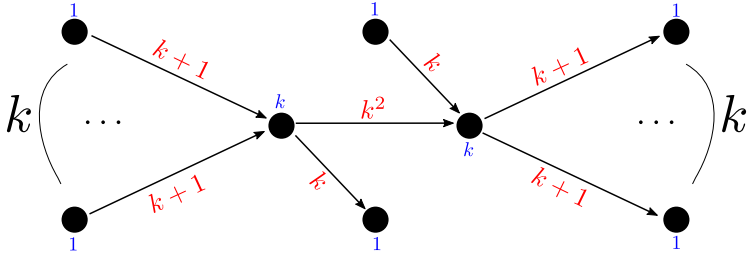}
     \caption{A DAG in $\ST_k$ that is not a  $k$-funnel, for every $k > 1$. The central edge is a forbidden path whose first vertex has indegree $k$ and outdegree $2$, and last vertex has indegree $2$ and outdegree $k$, the rest of the edges have either a source tail or a sink head. The blue label next to each vertex $v$ corresponds to $\min(\mu_s(v),\mu_t(v))$, since the maximum of these labels is $k$ the graph belongs to $\ST_k$. The red label next to each edge $e$ corresponds to $\mu(e)$, since there is a source-to-sink path with no $k$-private edge the graph is not a $k$-funnel.}
     \label{fig:strict-containment}
 \end{figure}

By noting that the minimum $k$ such that a DAG is in $\source_k$, $\T_k$ and $\ST_k$ is $\max_{v\in V} \mu_s(v)$, $\max_{v\in V} \mu_t(v)$ and $\max_{v\in V} \min(\mu_{s}(v), \mu_{t}(v))$, respectively, we obtain the same results as in \Cref{thm:k-funnel-linear-recognition,cor:exponential-min-k,lemma:linear-min-k} (with analogous assumptions on the cost of arithmetic operations) for recognition of $\source_k$, $\T_k$ and $\ST_k$.

 Next, we prove that although the vertex partition characterization of funnels does not generalizes to $k$-funnels, it does for the class $\ST_k$ and it can be found efficiently.
 
 \begin{lemma}\label{thm:stk-partitioning}
 Let $G = (V, E) \in \ST_k$ and $k$ given as inputs. We can find, in $O(|V|+|E|)$ time, a partition $V = V_1 \dot\cup V_2$ such that $G[V_1]\in \source_k$, $G[V_2]\in \T_k$ and there are no edges from $V_2$ to $V_1$. Moreover, if such a partition of a DAG $G$ exists, then $G \in \ST_k$.
 \end{lemma}
 \begin{proof}
 We set $V_1 = \{v \in V \mid \mu_s(v) \le k\}$ and $V_2 = V \setminus V_1$. Note that finding $V_1$ takes linear time, since we can apply the algorithm described in \Cref{thm:k-funnel-linear-recognition} to compute the $\mu_s$ values (or decide that they are more than $k$). By construction we know that every $v \in V_1$ has $\mu_s(v) \le k$, and since $G \in \ST_k$ also every $v \in V_2$ has $\mu_t(v) \le k$, thus $G[V_1]\in \source_k$, $G[V_2]\in \T_k$. Suppose by contradiction that there exists $e = (u, v)\in E \cap (V_2\times V_1)$. As such, $\mu_s(u) > k$, but since $\mu_s(u) \le \mu_s(v)$, then $\mu_s(v) > k$, a contradiction. Finally, if such a partition exists then $\mu_s(v) \le k$ for every $v \in V_1$ and $\mu_t(v) \le k$ for every $v \in V_2$, and thus $G\in\ST_k$.
 \end{proof}
 
 \section{Parameterized algorithms: The DAG}\label{sec:the-dag}
 
 The main idea to get the parameterized algorithms in this section is to bound the size of the $PI_v$ sets by a topological graph parameter and use \Cref{lemma:parameterized-vertex,cor:linear-parameterized-vertex} to obtain a parameterized solution. As in the KMP algorithm~\cite{knuth1977fast} only one prefix-incomparable value suffices (the longest prefix match until that point), we show that $\mu_s(v)$ prefix-incomparable values suffice to capture the prefix matches up to $v$.
 
 \begin{lemma}\label{lemma:prefix-incomparable-topo-bound}
 Let $G = (V, E)$ be a DAG, $v \in V$, $\paths_{sv}$ the set of source-to-$v$ paths, $\ell: V \rightarrow \Sigma$ a labeling function, $S \in \Sigma^m$ a string, and $PI_v$ as in \Cref{def:bv-piv}. Then, $|PI_v| \le \mu_s(v)$.
 \end{lemma}
 \begin{proof}
  Since any path ending in $v$ is the suffix of a source-to-$v$ path we can write $B_v$ as:
  \begin{align*}
      B_v &= \bigcup_{P_{sv} \in \paths_{sv}} B_{P_{sv}}:= \{i \in \{0,\ldots,m\} \mid \exists P \text{ suffix of }P_{sv}, \ell(P) = S[1..i]\}
  \end{align*}
  
  However, for every pair of values $i < j \in B_{P_{sv}}$, $S[1..i]$ is a border of $S[1..j]$ (it is a suffix since they are both suffixes of $\ell(P_{sv})$). As such, at most one value of $B_{P_{sv}}$ appears in $PI_v$, and then $|PI_v| \le |\paths_{sv}| = \mu_s(v)$.
 \end{proof}
 
 This result directly implies a parameterized string matching algorithm to DAGs in $\source_k$.
 
 \begin{lemma}\label{thm:dag-algorithm-dag-1}
 Let $G = (V, E) \in \source_k$, $\ell: V \rightarrow \Sigma$ a labeling function and $S \in \Sigma^m$ a string. We can decide whether $S$ has a match in $G,\ell$ in time $O(|V|k+|E| + \sigma m)$.
 \end{lemma}
 \begin{proof}
 We proceed as in \Cref{thm:dag-algorithm-pattern}, but instead we compute each $PI_v$ with the $O(k_v)$ version of \Cref{cor:linear-parameterized-vertex}. The claimed running time follows since $k_v = \sum_{u \in N^{-}_{v}} |PI_u| \le \sum_{u \in N^{-}_{v}} \mu_s(u) \le 1 + \mu_s(v) \le k+1$, by~\Cref{lemma:prefix-incomparable-topo-bound},~\Cref{eq:s2s-counting} and since $G \in \source_k$. 
 \end{proof}
 
 A simple but interesting property about string matching to graphs is that we obtain the same problem by reversing the input (both the graph and the string), that is, $S$ has a match in $G, \ell$ if and only if $S^r$ has a match in $G^r, \ell$. This fact, plus noting that $G\in \source_k$ if and only if $G^r \in \T_k$ gives the following corollary of \Cref{thm:dag-algorithm-dag-1}.
 
 \begin{corollary}\label{cor:dag-algorithm-dag-1}
  Let $G = (V, E) \in \T_k$, $\ell: V \rightarrow \Sigma$ a labeling function and $S \in \Sigma^m$ a string. We can decide whether $S$ has a match in $G,\ell$ in time $O(|V|k+|E| + \sigma m)$.
 \end{corollary}
 
 
 With these two results and the fact that we can compute the minimum $k$ such that a DAG is in $\source_k, \T_k$ in time $O((|V|+|E|)\log{k})$ (see \Cref{cor:exponential-min-k}) we obtain our first algorithm parameterized by the topology of the DAG.

\paramAlgTopoOne*
 
 Our final result is a parameterized algorithm for DAGs in $\ST_k$ (in particular for $k$-funnels). We note that the algorithm of \Cref{cor:dag-algorithm-dag-1} computes $PI_v$ for $S^r$ for every vertex in $G^r$. Recall that $PI_v$ represents all the prefix matches of $S^r$ with paths ending in $v$ in $G^r$. In other words, it represents all suffix matches of $S$ with paths starting in $v$ in $G$. For clarity, let us call this set $SI_v$. The main idea of the algorithm for $\ST_k$ is to use \Cref{thm:stk-partitioning} to find a partitioning $V = V_1 \dot\cup V_2$ into $\source_k$ and $\T_k$, use \Cref{thm:dag-algorithm-dag-1,cor:dag-algorithm-dag-1} to search for matches within each part and also to compute $PI_v$ for every $v \in V_1$ and $SI_v$ for every $v \in V_2$, and finally, to find matches using the edges from $V_1$ to $V_2$. The last ingredient of our algorithm consists of preprocessing the answers to the last type of matches.
 
 \begin{lemma}\label{lemma:prefix-suffix}
 Let $G = (V, E)$ a DAG, $(u, v) \in E$, $\ell: V \rightarrow \Sigma$ a labeling function, $S \in \Sigma^m$ a string and $PI_u$ and $SI_v$ as in \Cref{def:bv-piv}. We can decide if there is a match of $S$ in $G,\ell$ using $(u,v)$ in $O(|PI_u|\cdot|SI_v|)$ time, after $O(m^2)$ preprocessing time.
 \end{lemma}
 \begin{proof}
 We precompute a boolean table $PS$ of $m\times m$ entries, such that $PS[i,j]$ is \texttt{true} if there is a length $i'$ of a (non-empty) border of $S[1..i]$ (or $i' = i$) and a length $j'$ of a (non-empty) border of $S[m-j+1..m]$ (or $j' = j$) such that $i'+j'=m$, and \texttt{false} otherwise. This table can be computed by dynamic programming in $O(m^2)$ time as follows.
 \begin{align*}
     PS[i,j] &= \begin{cases}
\texttt{false} & \text{if }i+j < m \lor i = 0 \lor j = 0\\
i+j = m \lor PS[i, f_{S^r}(j)] \lor PS[f_S(i), j] & \text{otherwise}\end{cases}
 \end{align*}
 We then use this table to test every $PS[i,j]$ with $i \in PI_u, j\in SI_v$ and report a match if any of these table entries is \texttt{true}, in total $O(|PI_u|\cdot|SI_v|)$ time.
 
 Since every match of $S$ using $(u, v)$ must match a prefix $S[1..i]$ with a path ending in $u$ and a suffix $S[i+1..m]$ with a path starting in $v$, the previous procedure finds it (if any).
 \end{proof}

 \paramAlgTopoTwo*
 \begin{proof}
 We first compute the minimum $k$ such that the input DAG is in $\ST_k$ in time $O((|V|+|E|)\log{k})$ (see \Cref{cor:exponential-min-k}). Then, we obtain the partition of $G$ into $G[V_1] \in \source_k, G[V_2] \in \T_k$ and no edges from $V_2$ to $V_1$. We then search matches within $G[V_1]$ and $G[V_2]$ in time $O(|V|k+|E| + \sigma m)$ (\Cref{thm:dag-algorithm-dag-1,cor:dag-algorithm-dag-1}) and we also keep $PI_u$ for every $u\in V_1$ and $SI_v$ for every $v\in V_2$. Finally, we process the matches using the edges $(u, v)$ with $u\in V_1, v\in V_2$ in total $O(|E|k^2 + m^2)$ time (\Cref{lemma:prefix-suffix}) since $O(|PI_u|\cdot|SI_v|) = O(k^2)$.
 \end{proof}
  
 \section{Conclusions}
 In this paper we introduced the first parameterized algorithms for matching a string to a labeled DAG, a problem known (under SETH) to be quadratic even for a very special type of DAGs. Our parameters depend on the structure of the input DAG.
 
 We derived our results from a generalization of KMP to DAGs using prefix-incomparable matches, which allowed us to bound the running time to parameterized linear. Further improvements on the running time of our algorithms remain open: is it possible to get rid of the automaton? or to combine prefix-incomparable and suffix-incomparable matches in better than quadratic (either in the size of the sets or the string)? (e.g. with a different tradeoff between query and construction time of the data structure answering these queries) and is there a (conditional) lower bound to combine these incomparable sets? (see e.g.~\cite{bernardini2019even}). Another interesting question with practical importance is whether our parameterized approach can be extended to string labeled graphs with (unparameterized) linear time in the total length of the strings or extended to counting and reporting algorithms in linear time in the number of occurrences. 
 
 We also presented novel algorithmic results on funnels as well as generalizations of them. These include linear time recognition algorithms for their different characterizations, which we showed useful for the string matching problem but hope that can also help in other graph problems. We also showed how to find the minimum $k$ for which a DAG is a $k$-funnel or $\in \ST_k$ (assuming constant time arithmetic operations on numbers up to $O(k)$) using an exponential search, but it remains open whether there exists a linear time solution.
 


\bibliography{references}

\newpage
\appendix

 \section{A parameterized algorithm: The String}\label{sec:the-pattern}
 
  A simple property about prefix-incomparable sets is that their sizes are bounded by the number of prefixes that are not a border of other prefixes of the string, equivalently, the number of leaves in the failure function of the string.
 
 \begin{lemma}\label{lemma:bounded-size-prefix-incomparable}
 Let $S \in \Sigma^m$ be a string, $f_S$ its failure function/tree, and $B \subseteq \{0,\ldots,m\}$ prefix-incomparable for $S$. Then, $|B| \le w$ such that $w$ is the number of leaves of $f_S$, equivalently $w := |\{i \in \{0,\ldots, m\} \mid \not \exists j, f_S(j) = i\}|$.
 \end{lemma}
 \begin{proof}
 First note that $i < j$ are prefix-incomparable if and only if $i$ is not ancestor of $j$ in $f_S$. Suppose by contradiction that $|B| > w$, and consider the $w$ leaf-to-root paths of $f_S$. Note that these $w$ leaf-to-root paths cover all the vertices of $f_S$. By pigeonhole principle, there must be $i < j \in B$ in the same leaf-to-root path, that is $i$ is ancestor of $j$, a contradiction.
 \end{proof}

\begin{restatable}{theorem}{dagAlgorithmPattern}\label{thm:dag-algorithm-pattern}
Let $G = (V, E)$ be a DAG, $\Sigma$ a finite ($\sigma = |\Sigma|$) alphabet, $\ell: V \rightarrow \Sigma$ a labeling function, $S \in \Sigma^m$ a string and $f_S$ its failure function. We can decide whether $S$ has a match in $G,\ell$ in time $O((|V|+|E|)w + \sigma m)$, where $w = |\{i \in \{0,\ldots, m\} \mid \not \exists j, f_S(j) = i\}|$.
\end{restatable}
 \begin{proof}
  We compute the matching automaton $A_S$ in $O(\sigma m)$ time. Then, we process the vertices in topological order, and for each vertex $v$ we compute $PI_v$, the unique prefix-incomparable set representing $B_v$ (all prefix matches of $S$ with paths ending in $v$). We proceed according to \Cref{lemma:parameterized-vertex,cor:linear-parameterized-vertex} in $O(m)$ preprocessing time plus $O(w\cdot d^{-}_{v})$ time per vertex, adding up to $O(w(|V|+|E|))$ time in total. There is a match of $S$ in $G, \ell$ if and only if any $PI_v$ contains $m$.
 \end{proof}

 We note that $w \le m$, thus our algorithm is asymptotically as fast as the DAG algorithm, which runs in time $\Omega((|V|+|E|)m)$. However, we note that for $w$ to be $o(m)$, $S$ must be a periodic string. To see this, consider the longest prefix $S[1..i]$ of $S$, such that there exists $j > i$ with $i = f_{S}(j)$. By definition, $S[1..i]$ is a border of $S[1..j]$, thus $S[k] = S[k+j-i]$ for $k \in \{1, \ldots, i\}$, that is $S[1..j]$ is a periodic string with period $j-i$. Finally, note that if $w = o(m)$, then $m-i \in o(m)$, and thus the period $j-i \in o(m)$.

 \section{A linear time parameterized algorithm for the distance problems}\label{sec:linear-distance}

  Millani et al.~\cite{millani2020efficient} gave a $O(|V|(|V|+|E|))$ time algorithm to find a minimal forbidden path in a general graph. They used this algorithm to design branching algorithms (see e.g.~\cite{cygan2015parameterized}) for the problems of finding maximum sized sets $V' \subseteq V, d_{v} := |V'|$ and $E' \subseteq E, d_{e} := |E'|$, such that $G[V']$ and $(V, E')$ are funnels, known as vertex and edge distance to a funnel. It is know that (unless P = NP) there is an $\epsilon > 0$ such that there is no polynomial time $|V|^{\epsilon}$ approximation~\cite{lund1993approximation} for the vertex version nor $(1+\epsilon)$ approximation~\cite{millani2020efficient} for the edge version. The authors~\cite{millani2020efficient} noted that if we consider a minimal forbidden path $P$ of $G$ of length $|P| > 1$, then the edges of $P$ can be contracted until $|P| = 1$ without affecting the size of the solution. Moreover, they noted that if we consider such $P$, two in-neighbors of the first vertex and two out-neighbors of the last\footnote{This structure is known as a \emph{butterfly}.}, then $V'$ must contain at least one of those $6$ vertices and $E'$ one of those $5$ edges, deriving $O(6^{d_{v}}|V|(|V|+|E|))$ and $O(5^{d_{e}}|V|(|V|+|E|))$ time branching algorithms for each problem\footnote{After removing forbidden paths all cycles are vertex-disjoint thus the rest of the problem can be solved by removing one vertex (edge) per cycle in one $O(|V|+|E|)$ time traversal.}~\cite[Corollary 1]{millani2020efficient}. The authors also developed a more involved branching algorithm, only for the edge distance problem on DAG inputs, running in time $O(3^{d_{e}}(|V|+|E|))$~\cite[Theorem 4]{millani2020efficient}.
 
 By noting that minimal forbidden paths can be further contracted to length zero (one vertex) in the vertex distance problem, and that a minimal forbidden path can be found in time $O(|V|+|E|)$ (\Cref{lemma:unitig}) we obtain the following result.

\begin{restatable}{theorem}{linearDistance}\label{thm:linear-distance}
 Let $G = (V, E)$ be a graph. We can compute the vertex (edge) deletion distance to a funnel in time $O(5^d(|V|+|E|))$, where $d$ is the deletion distance.
 \end{restatable}

 \begin{proof}
 We follow the branching approach as in~\cite[Corollary 1]{millani2020efficient}, but in the case of vertex distance we further contract the forbidden paths to length $0$, the correctness of this step follows by noting that any solution containing two different vertices in a forbidden path is not minimum, since we still get a funnel by removing one of them (from the solution). As such, the number of recursive calls is $\le 5$ for both problems. Moreover, by \Cref{lemma:unitig}, we can find a minimal forbidden path in time $O(|V|+|E|)$.
 \end{proof}
\end{document}